\documentclass[aps,pra,twocolumn,superscriptaddress]{revtex4-2}
\usepackage{amsmath}
\usepackage{amsthm}
\usepackage{amssymb}
\usepackage{mathtools}
\DeclarePairedDelimiter{\norm}{\lVert}{\rVert}%
\DeclarePairedDelimiter{\bra}{\langle}{\rvert}%
\DeclarePairedDelimiter{\ket}{\lvert}{\rangle}%
\DeclarePairedDelimiter{\expval}{\langle}{\rangle}%
\DeclarePairedDelimiter{\abs}{\lvert}{\rvert}%
\DeclarePairedDelimiterX\innerp[2]{\langle}{\rangle}{#1\delimsize\vert\mathopen{}#2}%
\DeclarePairedDelimiterX\braket[2]{\langle}{\rangle}{#1\delimsize\vert\mathopen{}#2}%
\DeclarePairedDelimiterX\braketOP[3]{\langle}{\rangle}{#1\,\delimsize\vert\,\mathopen{}#2\,\delimsize\vert\,\mathopen{}#3}%
\DeclarePairedDelimiterX\ketbra[2]{\lvert}{\rvert}{#1\delimsize\rangle\!\delimsize\langle#2}%
\DeclarePairedDelimiterX\outerp[2]{\lvert}{\rvert}{#1\delimsize\rangle\!\delimsize\langle#2}%
\DeclarePairedDelimiterX\projector[1]{\lvert}{\rvert}{#1\delimsize\rangle\!\delimsize\langle#1}%
\DeclareMathOperator{\tr}{tr}%
\usepackage[colorlinks]{hyperref}
\usepackage{xcolor}
\usepackage{tikz}
\usepackage{mathdots}
\usepackage{yhmath}
\usepackage{cancel}
\usepackage{color}
\usepackage{siunitx}
\usepackage{array}
\usepackage{multirow}
\usepackage{amssymb}
\usepackage{gensymb}
\usepackage{tabularx}
\usepackage{extarrows}
\usepackage{booktabs}
\usepackage{colortbl}
\usetikzlibrary{fadings}
\usetikzlibrary{patterns}
\usetikzlibrary{shadows.blur}
\usetikzlibrary{shapes}
\usepackage{mathrsfs}
\usepackage{qcircuit}
\usepackage[ruled]{algorithm2e}
\usepackage{diagbox}
\usepackage{hhline}
\usepackage{makecell}

\newtheorem{theorem}{Theorem}
\newtheorem{corollary}{Corollary}
\newtheorem{lemma}{Lemma}

\newtheorem{conjecture}{Conjecture}
\theoremstyle{definition}
\newtheorem{definition}{Definition}
\DeclareMathOperator{\poly}{poly}

\begin{document}
\title{Efficient quantum tomography of a polynomial subspace}

\author{Yat Wong}
\email{yatwong@uchicago.edu}
\affiliation{Pritzker School of Molecular Engineering, University of Chicago, Chicago, Illinois 60637, USA}

\author{Ming Yuan}
\affiliation{Pritzker School of Molecular Engineering, University of Chicago, Chicago, Illinois 60637, USA}

\author{Kevin He}
\affiliation{James Franck Institute, University of Chicago, Chicago, Illinois 60637, USA}
\affiliation{Department of Physics, University of Chicago, Chicago, Illinois 60637, USA}

\author{Srivatsan Chakram}
\affiliation{Department of Physics and Astronomy, Rutgers University, Piscataway, NJ 08854, USA}

\author{Alireza Seif}
\altaffiliation{Present address: IBM Thomas J. Watson Research Center, Yorktown Heights, NY 10598, USA}
\affiliation{Pritzker School of Molecular Engineering, University of Chicago, Chicago, Illinois 60637, USA}

\author{David I. Schuster}
\affiliation{Pritzker School of Molecular Engineering, University of Chicago, Chicago, Illinois 60637, USA}
\affiliation{James Franck Institute, University of Chicago, Chicago, Illinois 60637, USA}
\affiliation{Department of Physics, University of Chicago, Chicago, Illinois 60637, USA}
\affiliation{Department of Applied Physics, Stanford University, Stanford, California 94305, USA}

\author{Liang Jiang}
\affiliation{Pritzker School of Molecular Engineering, University of Chicago, Chicago, Illinois 60637, USA}

\date{\today}

\begin{abstract}
Quantum tomography is crucial for characterizing the quantum states of multipartite systems, but its practicality is often limited by the exponentially large dimension of the Hilbert space. Most existing approaches, such as compressed sensing and tensor network-based tomography,  impose structural constraints on the state to enable more resource-efficient characterization. However, not all physical states can be well-approximated with highly structured states. Here, we develop a partial quantum tomography method based on direct fidelity estimation (DFE) that focuses on a neighborhood subspace---the subspace spanned by states physically close to a given target state. Using this generalized DFE method, we estimate elements of the density operator within this subspace in a self-verifying manner. We investigate the efficiency of this approach under different sets of available measurements for various states and find that the set of available measurements significantly impacts the cost of DFE. For example, we show that Pauli measurements alone are insufficient for performing efficient DFE on all product states, whereas the full set of product measurements is sufficient. This method can be applied in many situations, including characterizing quantum systems with confined dynamics and verifying preparations of quantum states and processes.
\end{abstract}

\maketitle
\section{Introduction}

Due to the exponentially large dimension of its Hilbert space, a complete characterization of an arbitrary state in a multipartite quantum system is intractable: for an $m$-partite system, with each part having $n$ local dimensions, the Hilbert space dimension is $n^m$. In addition, noise during evolution can decohere the quantum system and lead to a mixed state in the end. Therefore, we must generally use density matrices to describe the states, further increasing the number of parameters that must be estimated. Hence, a full tomography requires determining $n^{2m}-1$ real parameters, which means that the number of identical copies of the system needed to be prepared grows exponentially as the number of parties increases. Therefore, the tomography task is usually impractical for large $m$ due to resource consumption. 

To be more resource-efficient for state characterization, additional constraints on the states' form or the system's structure are needed. For example, if the density matrices have a low rank, compressed sensing techniques can be applied to reduce sample complexity ~\cite{gross_quantum_2010,flammia_quantum_2012}. Efficient methods to do tomography using matrix product states (MPS) \cite{cramer_efficient_2010, lanyon_efficient_2017} or matrix product operators (MPO) \cite{baumgratz_scalable_2013, baumgratz_scalable_2013-1}, which are often related to 1D systems with local interaction, have also been developed. It is also possible to use an adaptive method \cite{huszar_adaptive_2012, quek_adaptive_2021} or even machine learning protocols \cite{torlai_neural-network_2018, tiunov_experimental_2020} to optimize the measurement basis chosen.

There are also approaches that avoid tomography completely: If preparation of the desired state is assumed to be as accurate as possible, one can apply direct fidelity estimation (DFE) with the ideal state, using an importance sampling rule~\cite{flammia_direct_2011, da_silva_practical_2011} to verify the state rather than performing a full tomography. However, encountering vanishing weights in the ``importance-weighting" rule introduced in previous works can lead to a significant increase in sampling overhead. Addressing this issue requires the use of cutoffs, which in turn introduces systematic errors. Moreover, previous methods are limited to estimating the overlap with pure states.

In many situations, the physical state is expected to exist in a confined subspace whose dimension scales polynomially rather than exponentially with the number of parties. For example, the state could be in a subspace where engineered dissipation or a blockade Hamiltonian constrains the dynamics~\cite{albert_pair-cat_2019,chakram_multimode_2022, yuan2023universalcontrolbosonicsystems}. Additionally, if the state preparation noise is small, one would expect the resultant state to be within a neighborhood of the target state, which we formalize later. However, states in such a subspace could be highly entangled and difficult to characterize or verify. 

To address these challenges, we establish a theoretical framework for performing tomography based on DFE in a polynomial-dimensional neighborhood subspace called Direction Extraction of Density Matrix Elements from Subspace Sampling Tomography (DEMESST). Here, a neighborhood subspace refers to a subspace where all basis states are ``close" to a certain pure state via some easy-to-implement unitaries. As a consequence of being based on DFE, DEMESST can \emph{self-verify} this polynomial neighborhood subspace assumption by measuring the population of the physical state in the target subspace. A simple version of DEMESST was recently demonstrated by He et al. \cite{He2024}, and here we present the details and generalizations that go beyond that work.

\begin{table*}[t]
\begin{tabular}{|c||*{4}{c|}}
\hline
\diagbox{Measurements}{Maximum DFE Cost}{States} & \makecell{Stabilizer \\ State \\ (Corollary \ref{cr:2})} & \makecell{Product \\ State \\ (Corollary \ref{cr:4})} & \makecell{Matrix Product State \\ (bond dimension $k=O(1)$) \\ (Corollary \ref{cr:5})} & \makecell{Arbitrary Pure \\ State \\ (Corollary~\ref{cr:1})} \\ \hline & \multicolumn{4}{c|}{} \\[-3mm] \hline
Pauli Measurements & \multirow{4}{*}[-3mm]{$\leq2-2^{1-m}$} & \multicolumn{3}{c|}{$2^{\Theta(m)}$} \\ \cline{0-0} \cline{3-5}
Product Measurements &&& \multicolumn{2}{c|}{} \\ \cline{0-0} \cline{4-4}
\makecell{Product Measurements with \\ $O(m)$ Quasilocal Gates and \\ $O(1)$ Ancilla Qubits} && \multicolumn{2}{c|}{$1$} & \textcolor{blue}{$2^{O(n)}$}\\ \cline{0-0} \cline{5-5}
Arbitrary Measurements && \multicolumn{3}{c|}{} \\ \hline & \multicolumn{4}{c|}{} \\[-3mm] \hline
LOCC (if conjecture \ref{cj:1} holds) & \multicolumn{4}{c|}{\textcolor{blue}{$\leq 2-2^{1-m}$}} \\ \hline
\end{tabular}
\caption{Summary of DFE cost $Z_\mathcal{M}$ of multiqubit states with sets of measurements $\mathcal{M}$. A value of $Z_\mathcal{M}=1$ indicates that the projection of the state can be implemented perfectly within $\mathcal{M}$, while $2-2^{1-m}$ is usually achieved by a set of traceless measurements such that the (weighted) average is proportional to the projector of the state, up to a constant deviation. Blue entries indicate that we do not have tight proven bounds yet: $2^{O(n)}$ is just a trivial upper bound, while $2-2^{1-m}$ for arbitrary pure state with LOCC measurements depends on a conjecture. Note that a state do not need to be strictly in a category for it to be in the neighbourhood of such category, refer to Section \ref{sec:2} and Theorem \ref{th:1} for details. \label{table:1}}

\end{table*}
\begin{table*}[t]
\begin{tabular}{|c||*{5}{c|}}
\hline
\diagbox[innerwidth=13em,height=4\line]{Measurements}{Maximum \\ DFE Cost}{States} & \makecell{GKP \\ State \\ (Corollary \ref{cr:3})} & \makecell{Coherent \\ State} & \makecell{Product \\ State \\ (Corollary \ref{cr:4})}& \makecell{Pure \\ Gaussian \\ State \\ (Corollary \ref{cr:6})}& \makecell{Arbitrary Pure \\ State \\ (Corollary~\ref{cr:1})}\\ \hhline{|=#=:=:=:=:=|}
\makecell{Parity Measurements with \\ Displacement} & \multirow{2}{*}[-4mm]{$2$} & \multicolumn{1}{c}{} & \multicolumn{3}{c|}{\multirow{2}{*}{$2^{\Omega(m)}$}}\\ \cline{0-0} \cline{3-3}
\makecell{Parity Measurements and \\ Vacuum Projections with \\ Displacements} &&&\multicolumn{3}{c|}{}\\ \cline{0-0} \cline{4-6}
Product Measurements && \multicolumn{2}{c|}{} & \multicolumn{2}{c|}{} \\ \cline{0-1} \cline{5-5}
\makecell{Product Measurements with \\ Gaussian Unitaries} & \multicolumn{4}{c|}{$1$} & \textcolor{blue}{$?$}\\ \cline{0-0} \cline{6-6}
Arbitrary Measurements & \multicolumn{5}{c|}{} \\ \hhline{|=#=:=:=:=:=|}
LOCC (if conjecture \ref{cj:1} holds) & \multicolumn{5}{c|}{\textcolor{blue}{$\leq2$}} \\ \hline
\end{tabular}
\caption{Summary of DFE cost of multimode bosonic states. $1$ indicates that the projector can be achieved directly, while $2$ is usually achieved with a (weighted) average of traceless measurements. Once again, blue entries indicate we do not have tight proven bounds yet.\label{table:2}}

\end{table*}

Under such a framework, the set of available measurements decides what states' neighborhoods permit efficient subspace tomography. For example, as we show later, there is an exponential gap in sampling overhead between Pauli measurements and general product measurements when performing tomography of a neighborhood subspace of a tensor product of magic states.

In the following, we first define the polynomial neighborhood space, which is spanned by states that differ from a pure state by some physically relevant operations. We also establish Generalized DFE to address issues of the existing method of DFE: while previous proposals were limited to pure state projections, our method is applicable to arbitrary Hermitian operators. Additionally, we introduce a sampling scheme that does not suffer from the small denominator issue that stems from vanishing weights in the existing schemes. Our method improves the variance per measurement and goes beyond the previously considered Pauli and Wigner measurements~\cite{flammia_direct_2011, da_silva_practical_2011}. After that, we determine a sufficient condition for performing tomography within the subspace efficiently using DEMESST. We explore sets of measurements that permit efficient DFE of certain classes of states. Finally, we identify states with exponential gaps in the DFE overhead between similar sets of measurements, such as Pauli measurements and product measurements. 
\section{Summary of results\label{sec:2}}
The results are summarized in TABLES \ref{table:1} \& \ref{table:2}, representing the maximum DFE overhead of certain families of states with certain measurements. Each column represents the base state of a neighborhood subspace, while each row represents the available measurements. Diagonal entries are proved in the following sections as corollaries. If an entry is constant, then the neighborhood generated from a base state of that category with local operators allows efficient tomography with the corresponding set of measurements. Note that a state $\sigma$ does not need to be in the category itself to allow efficient neighborhood tomography. For example, a $W$ state itself is not a product nor a stabilizer state. However, being the equal superposition of the states resulting from applying $X$ of different qubits on $\ket{0}^{\otimes m}$, it is considered to be in the neighborhood of the $\ket{0}^{\otimes m}$ state. Since $\ket{0}^{\otimes m}$ is a product state and a stabilizer state, it has a constant overhead, as shown in the table. Hence, the DFE overhead of a $W$ state is at most polynomial in $m$ for all listed measurement sets.
The second column of TABLE \ref{table:2} has been experimentally demonstrated by performing tomography in a bounded photon subspace in up to four modes of a multimode bosonic system~\cite{He2024}.
\section{Polynomial neighborhood subspace}
In many situations, we expect the state prepared in an experiment to only slightly differ from a target pure state. To formalize this, we define the neighborhood subspace of a pure state as follows.
\begin{definition}
Given a pure state $\ket{\psi}$ and a set of operators $\mathcal{K}$, define the neighborhood of $\ket{\psi}$ generated by $\mathcal{K}$ $k$ times iteratively as
\begin{equation}
N_0(\ket{\psi},\mathcal{K})=\mathop{span}(\{\ket{\psi}\}),
\end{equation}
\begin{equation}
\begin{split}
N_{k+1}(\ket{\psi},\mathcal{K})=\mathop{span}(&\left\{K\ket{\psi'}\right|K\in\{I\}\cup\mathcal{K},\\
&\left.\ket{\psi'}\in N_k(\ket{\psi},\mathcal{K})\right\}).
\end{split}
\end{equation}
\end{definition}
For example, if we prepare some state $\ket{\psi}$ and allow it to evolve under some Lindbladian $\mathcal{L}(\rho)=\gamma\sum_i(L_i\rho L_i^\dagger-\frac{1}{2}\{L_i^\dagger L_i,\rho\})$ for a short time $t\ll\gamma^{-1}$, we would expect most of the state to end up within $N_k(\ket{\psi},\cup_i\{L_i, L_i^\dagger L_i\})$, i.e.
$$\tr \left(\Pi_k\rho_t\Pi_k\right)=1-O\left((\gamma t)^{k+1}\right),$$
where $\Pi_k$ is the projection operator onto $N_k$. Such spaces have dimensions polynomial in $\abs{\mathcal{K}}$ regardless of the dimension of the complete Hilbert space:
\begin{lemma}
The dimension of the neighborhood generated from a pure state $\ket{\psi}$ (``base state") with at most $k$ applications of a combination of operators from $\mathcal{K}$ is polynomial in $\abs{\mathcal{K}}$ if $k=O(1)$.
\end{lemma}
\begin{proof}
\begin{equation}
\dim\left[N_0(\ket{\psi},\mathcal{K})\right]=1,
\end{equation}
\begin{equation}
\begin{split}
\dim\left[N_k(\ket{\psi},\mathcal{K})\right]&\leq(\abs{\mathcal{K}}+1)\dim\left[N_{k-1}(\ket{\psi},\mathcal{K})\right]\\
&\leq(\abs{\mathcal{K}}+1)^k.
\end{split}
\end{equation}
\end{proof}
Hence, the exponential dimension of the entire Hilbert space does not inhibit partial tomography of the relevant subspace. As an example, the tomography of any polynomial dimensional subspace $\mathcal{S}$, which can be seen as a neighborhood subspace generated from a pure state $\ket{\psi}\in\mathcal{S}$ with $\dim(\mathcal{S})-1$ different unitaries, is possible if arbitrary measurements are allowed:
\begin{corollary}
\label{cr:1}
If arbitrary (multipartite) measurements are allowed, tomography of any subspace $\mathcal{S}$ within $\epsilon$ Frobenius distance with at least $1-\delta$ success probability can be performed in $\poly(\dim{\mathcal{S}},\epsilon,\ln\delta^{-1})$ measurements.
\end{corollary}
\begin{proof}
By choosing an orthonormal basis $\{\ket{\psi_l}\}$ of $\mathcal{S}$, we only need to treat the system as a $(\dim{\mathcal{S}})$-level qudit and perform standard tomography to retrieve the projected density operator. Note that the resultant operator should not be renormalized to unit trace because the physical state may not entirely lie within this subspace, and this procedure only retrieves the projected state. This is listed as the final column in TABLES \ref{table:1} \& \ref{table:2}.
\end{proof}
However, arbitrary measurements are hard to implement. Therefore, in the following, we present a method to realize this tomography with available measurements.

\section{Generalized direct fidelity estimation\label{sec:gdfe}}
To build a method for performing tomography in some $d$-dimensional subspace, which corresponds to estimating the expectations of a linearly independent set of $d^2$ Hermitian operators, we need the capability to effectively estimate the expectation of an arbitrary Hermitian operator in the subspace. DFE can be regarded as a special case with $d=1$, where the overlap of two states is sampled with certain families of measurements, such as Pauli measurements and Wigner measurements. For example, one can write the overlap of two states as the inner product of the vectors of expectation values of such measurable operators and estimate it by sampling the ratio of expected value with the reference value squared as the probability function ~\cite{flammia_direct_2011, da_silva_practical_2011}. In particular, to obtain the overlap between some known $m$-qubit pure state $\sigma=\ketbra{\psi}{\psi}$ and some unknown $m$-qubit physical state $\rho$, one can perform DFE to estimate $\tr{(\rho\sigma)}=2^{-m}\sum_{P\in\mathcal{P}}\tr{(P\rho)}\tr{(P\sigma)}=\sum_{P\in\mathcal{P}}p_P\frac{\tr{(P\rho)}}{\tr{(P\sigma)}}$, where $\mathcal{P}=\left\{I,X,Y,Z\right\}^{\otimes m}$ is a subset of the Pauli group, and $p_P=2^{-m}\left(\tr{(P\sigma)}\right)^2$. Note that $p_P$ does not depend on $\rho$ and is determined before the measurements. Such methods suffer from statistical and/or systematic errors induced by small denominators when some of the reference values have small magnitudes.

To address the vanishing denominator problem, we propose optimizing the choice of sampling probability and adjusting the weight of each measurement accordingly. For simplicity, we assume that the operator we measure is Hermitian, and all possible measurement operators have only measurement results of $\pm1$. We can express the expectation value of any Hermitian operator $\mathcal{O}$ within the subspace spanned by the identity and the available measurement operators, as a weighted sum of the expectation values of measurements: 
\begin{equation}
\expval{\mathcal{O}}_\rho=C(\mathcal{O})+\sum_i f_i(\mathcal{O})\tr(M_i\rho),
\end{equation}
where $M_i$ is the $i^{th}$ available measurement, $f_i$ and $C$ are real coefficients dependent on $\mathcal{O}$ and the set of measurement operators, and $\tr(M_i\rho)$ are the expectation values of measurement. In DFE, $\expval{\mathcal{O}}_\rho$ is estimated as
\begin{equation}
\mathbb{E}_i\left[w_i\expval{M_i}_\rho\right]=\sum_ip_iw_i\tr(M_i\rho),
\end{equation}
where $p_i$ is the probability of performing measurement $M_i$ and $w_i$ is the weight assigned to the measurement. If $p_iw_i=f_i$, the weighted mean is an unbiased estimator of the overlap, i.e. if we sample a random variable $X$ where we randomly select an $i$, each with a probability of $p_i$, measure $M_i$, and assign the measurement outcome multiplied by $w_i$ to $X$, then we have
\begin{equation}
\mathbb{E}\left[X\right]=\expval{\mathcal{O}-C}_\rho.
\end{equation}
For pure reference states $\mathcal{O}=\sigma$, traditional DFE methods~\cite{flammia_direct_2011, da_silva_practical_2011} have $C=0$, and use $p_i=f_i(\sigma)\tr(M_i\sigma)$ and $w_i=\left(\tr(M_i\sigma)\right)^{-1}$. This formulation of $p_i$ and $w_i$ utilizes the fact that $\sum_if_i(\sigma)\tr(M_i\sigma)=\tr(\sigma^2)=1$, which implies $p_i$ is a probability distribution. However, as discussed previously, this suffers from potentially divergent fractions, and the application of a cut-off to avoid statistical errors from small denominators induces a systematic error.

In our generalized DFE method, we minimize variance per measurement:
\begin{equation}
\begin{split}
& \text{Var}\left[X\right]=\mathbb{E}_i\left[w_i^2\right]-\expval{\mathcal{O}-C}_\rho^2=\sum_ip_i\abs{w_i}^2-\expval{\mathcal{O}-C}_\rho^2\\
= & \sum_ip_i \abs{w_i}^2 \sum_j p_j-\expval{\mathcal{O}-C}_\rho^2 \geq \left(\sum_i\abs{f_i}\right)^2-\expval{\mathcal{O}-C}_\rho^2.
\end{split}
\end{equation}
Due to the Cauchy-Schwarz inequality, the variance is minimized when the equality is satisfied, which requires $p_i\abs{w_i}^2\propto p_i$. Therefore, all $w_i$ must have the same magnitude, thus $p_iw_i=f_i$ implies $p_i\propto\abs{f_i}$. Hence, we construct the following distribution:
\begin{equation}
\left\{\begin{aligned}
& Z(\mathcal{O})=2\sum_j\abs{f_j(\mathcal{O})},\\
& w_i(\mathcal{O})=\frac{1}{2}\text{sgn}[f_i(\mathcal{O})]Z(\mathcal{O}),\\
& p_i(\mathcal{O})=\frac{2\abs{f_i(\mathcal{O})}}{Z(\mathcal{O})},\end{aligned}\right.
\end{equation}
where $Z$, twice the sum of all weights, is effectively the overhead of this procedure, as the variance per measurement is upper bounded by $Z^2/4$.Since such a distribution eliminates the possibility of small denominators, we eliminate the systematic errors generated by introducing a cutoff for handling small denominators in existing DFE methods. Note that while this approach attains minimal sample complexity over all possible parametrizations of $p_i$ and $w_i$ that satisfy $p_iw_i=f_i$, it does not rule out instances of states with exponential sample complexity.
Since the measurement outcomes can only be $\pm1$, Hoeffding's inequality provides us a bound on the sampling complexity that is polynomial in $Z$, the inverse of tolerable error $\epsilon$, and the logarithm of the failure rate $\delta$.

In principle, we can relax the measurement outcomes to arbitrary real numbers and consider implementations of measurements with POVMs. Note that with such relaxation, the sampling distribution is no longer guaranteed to be optimal in variance per measurement, as the proof relies on the property that the square of any measurement outcome is 1. If $M_i$ is implemented with POVM $\{\Lambda_j^{(i)}\}$, then we have
$$M_i=\sum_j\lambda_j^{(i)}\Lambda_j^{(i)},$$
where $\lambda_j^{(i)}$ indicates the measurement result assigned to $\Lambda_j^{(i)}$. Since there are multiple ways to write an operator as a weighted sum of measurements, we define $Z_\mathcal{M}(\mathcal{O})$, the ``$\mathcal{M}$-DFE scale factor", to be the minimal $Z$ over all possible ways of writing $\mathcal{O}$ as a weighted sum of measurements in $\mathcal{M}$. This is formally defined as the following:
\begin{equation}
\begin{aligned}
Z_\mathcal{M}(\mathcal{O}) &=\min_{\{M_i\}\subseteq\mathcal{M}}Z(\mathcal{O})\\
&=\min_{\substack{\{M_i\}\subseteq\mathcal{M},\{f_i\},C\\\text{s.t.}\sum_if_iM_i=\mathcal{O}-CI}}\sum_i\left(\lambda_{\max}^{(i)}-\lambda_{\min}^{(i)}\right)\abs{f_i},
\end{aligned}
\end{equation}
where $\lambda_{\max}^{(i)}$ and $\lambda_{\min}^{(i)}$ indicate the maximum and minimum of measurement result assignments of the POVM $\left\{\Lambda_j^{(i)}\right\}$ for achieving $M_i$. $Z(\mathcal{O})$ satisfies nonnegativity, homogeneity, and the triangle inequality, and hence can be interpreted as a norm. Note that when perfect projection to the state is available, i.e. $\projector{\psi}\in\mathcal{M}$, we have $Z_\mathcal{M}(\projector{\psi})=1$ because $M$ can be chosen to be $\projector{\psi}$. Here, the POVM is $\{\projector{\psi},I-\projector{\psi}\}.$ There are also situations where assigning different ranges for different measurements would improve the variance and hence the overall performance, but for simplicity, we only consider fixed ranges in the main text. The analysis for general POVMs and unbounded measurements is in Appendix \ref{sec:povm_unbounded}.

So far we have considered performing generalized DFE of arbitrary operators and characterized the efficiency. Next, we utilize this method to perform DFE-based tomography in a subspace and identify cases where such tomography can be done efficiently.
\begin{definition}
Define a state space $\mathcal{S}$ to have a cost of $\beta$ with $\mathcal{M}$ if and only if the following two conditions are satisfied.

\begin{equation}\dim{\mathcal{S}}\leq \beta,\end{equation}
\begin{equation}
\max_{\substack{\ket{\psi}\in\mathcal{S} \\ \braket{\psi}{\psi}=1}}  Z_\mathcal{M}(\projector{\psi})\leq \beta.
\end{equation}

\end{definition}
\begin{lemma}
A $\poly(m)$ dimensional space $\mathcal{S}$ has a $\poly(m)$ cost with $\mathcal{M}$ if and only if there exists an orthonormal basis $\{\ket{\psi_a}\}$ such that
$$\max_{1\leq a,b\leq\dim{\mathcal{S}},c\in\{0,1\}}Z_\mathcal{M}(i^c\ketbra{\psi_a}{\psi_b}+h.c.)=O\left(\poly(m)\right).$$
\end{lemma}
\begin{lemma}
Given a $\poly(m)$-cost space $\mathcal{S}$ with $\mathcal{M}$, $\rho_\mathcal{S}=P_\mathcal{S}\rho P_\mathcal{S}$ can be determined with probability $1-\delta$  within Frobenius distance $\epsilon$ by performing $\poly(m,\ln\delta,\epsilon^{-1})$ measurements from $\mathcal{M}$.
\end{lemma}
The proofs of lemmas are in Appendix \ref{proofs}. Lemma 2 and 3 establish that the capability of performing efficient DFE of any state in a polynomial subspace  guarantees efficient tomography of the projected density operator. For neighborhood subspaces, under certain conditions, one only needs to be capable of performing DFE of a base state $\ket{\psi}$ efficiently, because it guarantees that any states within the neighborhood subspace generated from this state will also have efficient DFE by expanding the measurements to include Hadamard tests with the generators of the neighborhood (see Fig. \ref{fig:1}). To state this efficiency condition of neighborhood subspace formally:

\begin{figure}[]

\includegraphics[]{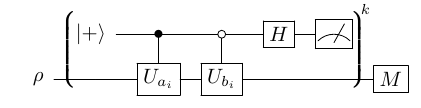}
\caption{An example of a measurement in the Hadamard-expanded set of measurements $\mathcal{M}$, generated by performing multiple Hadamard tests with $U\in\mathcal{U}$, followed by the measurement $M\in\mathcal{M}'$. This allows efficient estimation of $\tr{\left(U_a^\dagger\projector{\psi}U_b\right)}$.\label{fig:1}}
\end{figure}

\begin{theorem}\label{th:1}
Given a basis state $\ket{\psi}$ and a base set of measurements $\mathcal{M}'$, if $\mathcal{M}'$ permits efficient DFE of $\ket{\psi}$, i.e.
$$Z_{\mathcal{M}'}\left(\projector{\psi}\right)=O(\poly(m)),$$
then for $N_k(\ket{\psi},\mathcal{U})$, the neighborhood generated from $\ket{\psi}$ by up to $k$ applications of unitaries from $\mathcal{U}$, define the Hadamard-test-expanded set of measurements $\mathcal{M}$ to be
$$\mathcal{M} := \left\{\mathcal{C}(M)\mid M\in\mathcal{M}', \mathcal{C}\in\mathbb{O}\right\},$$
where $\mathbb{O}$ is the set of all sequences of controlled unitaries from $\mathcal{U}$ with ancilla qubits as control preceded by preparing such ancilla in $\ket{+}$ states followed by measuring these ancillae (See FIG.~\ref{fig:1}). If $k=O(1)$, $\abs{\mathcal{U}}=O(\poly(m))$, then
$$Z_{\mathcal{M}}\left(i^c\prod^k_{i=1}U_{a_i}\projector{\psi}\left(\prod^k_{j=1}U_{b_j}\right)^\dagger+h.c.\right)=O(\poly(m)),$$
where $U_{a_i},U_{b_i}\in\{I\}\cup\mathcal{U}$.

Furthermore, consider the Gram matrix $G_{\vec{a}\vec{b}}=\braketOP{\psi}{(\prod^k_{i=1}U_{a_i})^\dagger\prod^k_{j=1}U_{b_j}}{\psi}$ of the neighborhood basis $\left\{\prod^k_{j=1}U_{a_j}\ket{\psi}\right\}$, then
\begin{equation*}
\max_{\substack{\ket{\psi'}\in N_k(\ket{\psi},\mathcal{U})\\\braket{\psi'}{\psi'}=1}}Z_{\mathcal{M}}(\projector{\psi'})=O\left(\poly(m,\norm{G^+}_2)\right).
\end{equation*}
Here, $G^+$ is the Moore–Penrose inverse of $G$. In other words, $\norm{G^+}_2$ is the inverse of the smallest nonzero eigenvalue of $G$. Therefore, if $\norm{G^+}_2=O(\poly{(m)})$, then the DFE cost of any state in the neighborhood is $O(\poly{(m)})$.
\end{theorem}
The proof of Theorem 1 is deferred to the Appendix \ref{proofs}. In most physically relevant cases, $\norm{G^+}_2=O(1)$. In particular, if $\{\prod^k_{i=1}U_{a_i}\}$ is an orthonormal set, $G$ is the identity matrix, and hence $\norm{G^+}_2=1$. In this theorem, the crucial condition for the set of measurements to be sufficient for efficient DFE within the neighborhood subspace is that the measurement set is invariant under the Hadamard test. Most sets of measurements are invariant under Hadamard tests of local unitaries. For example, a product measurement preceded by any number of Hadamard tests of single qubit unitaries is still a product measurement. Thus, with these assumptions, any neighborhood around any DFE-efficient state generated by local operators allows for efficient DFE-based tomography. For example, if the set of possible measurements is the set of all product measurements, since local operations do not alter the product property of operators, the set of product measurements is invariant under local Hadamard tests. For Pauli measurements, although the Pauli group is not closed under arbitrary local operation, any local operation can only transform a single qubit Pauli operator into an operator with a 2-norm of at most 1, effectively rotating the axis of measurements. In the worst case, the new axis direction would be equidistant from the three principle axes, introducing a factor of $\sqrt{3}$ to the DFE cost. Hence, up to a constant $k$, such efficiency is still retained. The expansion of measurements with Hadamard tests can also be viewed as a method to expand the available set of POVM with controlled unitaries, for example, Pauli measurements are expanded to product measurements with local Hadamard tests, and product measurements in turn can be expanded to quasilocal measurements with Hadamard tests of two-qubit gates.

Finally, we present our DEMESST algorithm to perform tomography of a polynomial subspace as in Algorithm \ref{alg}. The resultant operator is an unbiased estimator of the projected density operator, and is not guaranteed to be positive semidefinite nor trace 1. In particular, the expected value of the trace is equal to the trace of the product of the physical state and the projection to the subspace. Hence, this procedure can self-verify: if the trace is close to one, the physical state is indeed mostly in such a subspace; if the state is mostly outside of the subspace, it will be reflected through the trace. If the application requires a positive semidefinite operator, one can project to the nearest physical state~\cite{PhysRevLett.108.070502}.
\begin{algorithm}
\caption{DEMESST. The inputs are the basis $\left\{\ket{\psi_j}\right\}$, the available set of measurements $\mathcal{M}$, and the sample count for the elements $t_{jl}$, while the output is an estimate of $\rho$ in the basis provided, i.e. $\rho_{jl}\sim\braketOP{\psi_j}{\rho}{\psi_l}$. For $j=l$, $t_{jj}$ indicates the sample count for the diagonal element $\rho_{jj}$. For $j<l$, $t_{jl}$ indicates the sample count for the real part of $\rho_{jl}$, while $t_{lj}$ indicates the sample count for the imaginary part.}\label{alg}
\KwData{$\left\{\ket{\psi_j}\right\},\mathcal{M},t_{jl}$}
\KwResult{$\rho_{jl}\sim\braketOP{\psi_j}{\rho}{\psi_l}$}
$\rho\gets \text{Zeros}(\abs{\left\{\ket{\psi_j}\right\}},\abs{\left\{\ket{\psi_j}\right\}})$\;
\For{$j\gets1$ to $\abs{\left\{\ket{\psi_j}\right\}}$}{
  Find $C$ and $f_i$ such that $\projector{\psi_j}=C+\sum_if_iM_i$\;
  \For{$a\gets1$ to $t_{jj}$}{
    Randomly sample and measure $M_i$ with $p_i=\frac{\abs{f_i}}{\sum_i\abs{f_i}}$, and assign the measurement result to $X_a$\;
  }
  $\rho_{jj}\gets$ Average of $\mathop{sgn}(f_{i_a})X_a$\;
  \For{$l\gets j+1$ to $\abs{\left\{\ket{\psi_j}\right\}}$}{
    Find $C$ and $f_i$ such that $\ketbra{\psi_j}{\psi_l}+h.c.=C+\sum_if_iM_i$\;
    \For{$a\gets1$ to $t_{jl}$}{
      Randomly sample and measure $M_i$ with $p_i=\frac{\abs{f_i}}{\sum_i\abs{f_i}}$ and assign the measurement result to $X_a$\;
    }
    $\rho_{jl}\gets$ Average of $\frac{1}{2}\mathop{sgn}(f_{i_a})X_a$\;
    Find $C$ and $f_i$ such that $i\ketbra{\psi_j}{\psi_l}+h.c.=C+\sum_if_iM_i$\;
    \For{$a\gets1$ to $t_{lj}$}{
      Randomly sample and measure $M_i$ with $p_i=\frac{\abs{f_i}}{\sum_i\abs{f_i}}$ and assign the measurement result to $X_a$\;
    }
    $\rho_{lj}\gets\rho_{jl}-$ Average of $\frac{i}{2}\mathop{sgn}(f_{i_a})X_a$\;
    $\rho_{jl}\gets\rho_{jl}+$ Average of $\frac{i}{2}\mathop{sgn}(f_{i_a})X_a$\;
  }
}
\end{algorithm}
\section{DFE Cost: Examples}
With this framework, we analyze classes of states with efficient DFE and the corresponding sets of operators that can generate $\poly(m)$-cost subspaces.\\

Firstly, we revisit the classical example of stabilizer states:
\begin{corollary}
\label{cr:2}
If $\ket{\psi}$ is a stabilizer state, tomography of the neighborhood space generated with polynomially many different Pauli operators, i.e. $N_k(\ket{\psi},\mathcal{U})$, where $\mathcal{U}\subseteq\mathcal{P}$, $\abs{\mathcal{U}}=O(\poly(m))$, and $k=O(1)$, can be done with a polynomial number of Pauli measurements.
\end{corollary}
The Pauli DFE cost factor of any Hermitian operator is
\begin{equation}
Z_\mathcal{P}\left(\mathcal{O}\right)=2^{1-m}\sum_{P\in\{I,\sigma_x,\sigma_y,\sigma_z\}^{\otimes m}}\abs{\tr\left(P\widetilde{\mathcal{O}}\right)}.
\end{equation}
For a stabilizer state, this would be $2^{1-m}(2^m-1)=2-2^{1-m}$, since within $2^m$ stabilizers, only $2^m-1$ are nontrivial. The neighborhood basis (or its subset, since if the product of two Paulis is proportional to a stabilizer, the two states will be identical up to a global phase) is orthonormal, hence the Gram matrix will have a bounded $\norm{G^+}_2$. Therefore, from Theorem \ref{th:1}, one can perform efficient tomography of the neighborhood subspace with Hadamard-expanded Pauli measurements. However, to prove the corollary, we need to go further and show that Pauli measurements are sufficient to perform the tomography efficiently. Hence, we consider the elements of the neighborhood basis directly. The off-diagonal elements in the neighborhood basis,
\begin{equation}
(P_i\ket{\psi})(\bra{\psi}P_j^\dagger)=2^{-m}\sum_{s\in S}P_isP_j^\dagger,
\end{equation}
is also an average of Pauli operators. After separating Hermitian and anti-Hermitian parts, we can sample it similarly to DFE of the stabilizer state, with $Z_\mathcal{P}\leq2-2^{1-m}$; therefore, the DFE of any element in the neighborhood basis and hence the tomography is efficient. Furthermore, since this process has no requirements beyond the stabilizer formalism, the generating Pauli operators can be non-local, such as single-qubit errors propagating through a Clifford circuit. Any Pauli error in a Clifford circuit would end up as a (potentially different) Pauli error in the final state. Hence, if we prepare a stabilizer state from a product state through a Clifford circuit with $\mathfrak{m}$ single or two-qubit gates, we can cover up to $k$ gate errors by considering the neighborhood generated by up to $k$ applications of the $\leq15\mathfrak{m}$ possible Pauli gate errors (15 corresponds to the number of nontrivial Paulis for two qubits), which has $O((15\mathfrak{m})^k)$ possibilities. This can be further extended to process tomography of circuits with $O(\ln{\mathfrak{m}})$ non-Clifford gates by representing non-Clifford gates as superpositions of Paulis and performing DFE of the Choi operator~\cite{PRXQuantum.5.030339}. Note that to perform DFE of the Choi operator, we do not need to physically prepare maximally entangled states with ancilla qubits. Instead, we can prepare product stabilizer states and measure the output state according to the Pauli operator being sampled on the Choi state. This is listed as the first column in TABLE \ref{table:1}.

Continuous variable systems have states analogous to stabilizer states in multiqubit systems, i.e. GKP states~\cite{GKP}, the eigenstates of certain commuting sets of displacement operators. For example, $\ket{\psi}=\left(\sum_{s=-\infty}^\infty \ket{q=\sqrt{2\pi}s}\right)^{\otimes m}$ is an $m$-mode square lattice GKP state. Such states are stabilized by displacement operators $D_\alpha$ and displaced parity operators $D_\alpha(-1)^{\sum\hat{n}}D_\alpha^\dagger$, as long as the displacement $\vec{\alpha}$ is a lattice vector of the state. Similar to qubit stabilizer states, DFE of such states also have a constant overhead with a simple set of measurements:
\begin{corollary}
\label{cr:3}
If $\projector{\psi}$ is an ideal GKP state and $\mathcal{M}$ consists of all Wigner measurements, i.e. $\expval{M(\vec{\alpha})}_\rho=\expval{D_\alpha(-1)^{\sum\hat{n}} D^\dagger_\alpha}_\rho=\left(\frac{\pi}{2}\right)^mW_\rho(\vec{\alpha})$, the DFE cost of $\projector{\psi}$ is $2$.
\end{corollary}
The overlap of two operators can be written as the overlap of the Wigner functions:
\begin{equation}
\begin{split}
\tr(\rho\sigma)&=\pi^m\int d^{2m}\vec{\alpha}\ W_\rho(\vec{\alpha})W_\sigma(\vec{\alpha})\\
&=2^m\int d^{2m}\vec{\alpha}\ \left(\left(\frac{\pi}{2}\right)^mW_\rho(\vec{\alpha})\right)W_\sigma(\vec{\alpha}),
\end{split}
\end{equation}
and the associated cost factor satisfies:
\begin{equation}
Z_{\mathcal{W}}(\sigma)=2^{m+1}\int d^{2m}\vec{\alpha}\ \abs{W_\sigma(\vec{\alpha})}.
\end{equation}
The Wigner function of a GKP state can only take values of $0$ or $\pm\left(\frac{2}{\pi}\right)^m$, therefore
\begin{equation}
\begin{split}
Z_{\mathcal{W}}(\projector{\psi})&=2^{m+1}\int d^{2m}\vec{\alpha}\ \abs{W_\psi(\vec{\alpha})}\\
&=2\pi^m\int d^{2m}\vec{\alpha}\ \left(W_\psi(\vec{\alpha})\right)^2\\
&=2.
\end{split}
\end{equation}
Equivalently, we are sampling each lattice point with equal probability. This is listed as the first column in TABLE \ref{table:2}.

Pauli and Wigner measurements can attain small overheads for some states, such as the examples above. However, there exist simple states that have exponential DFE costs with such measurements. For example, DFE of the tensor product of $m$ magic states
\begin{equation}
\projector{\psi}=\left(\frac{I+3^{-1/2}\left(\sigma_x+\sigma_y+\sigma_z\right)}{2}\right)^{\otimes m},
\end{equation}
with Pauli measurements has an associated cost factor $Z_\mathcal{P}$ that scales exponentially as the mode number grows:
\begin{equation}
Z_\mathcal{P}(\projector{\psi})=2^{1-m}\left(\left(1+\sqrt{3}\right)^m-1\right)=\Theta\left(\left(\frac{1+\sqrt{3}}{2}\right)^m\right).
\end{equation}
However, if we have access to arbitrary product measurements, including projection to arbitrary product states, this would allow for sampling the fidelity directly, and hence $Z=1$ if $\mathcal{M}$ is the set of all product measurements.

A simple example in multimode bosonic system that has an exponential DFE tomography cost with Wigner measurement is the multimode vacuum state, i.e.
\begin{equation}
\projector{\psi}=\projector{0}^{\otimes m}.
\end{equation}
Since $Z_\mathcal{W}\left(\mathcal{O}\right)=2^{m+1}\int d^{2m}\vec{\alpha}\abs{W_{\widetilde{\mathcal{O}}}(\vec{\alpha})}$,
\begin{equation}
Z_\mathcal{W}(\projector{\psi})=2^{m+1}\int d^{2m}\vec{\alpha}\left(\frac{2}{\pi}\right)^me^{-2\abs{\vec{\alpha}}^2}=2^{m+1}.
\end{equation}
However, if Husimi Q function measurements are allowed, including the product of projection to the vacuum of individual modes, the Q function at the origin will sample the fidelity, hence $Z_{\mathcal{Q}}=1$. Note that to perform neighborhood tomography around coherent states, a combination of $Q$ and $W$ measurements are often needed to effectively perform the Hadamard tests. This combination is the protocol used by He et al.~\cite{He2024}. However, the cost of DFE tomography of any normalizable state with only Wigner measurements $Z_W$, including any finite energy GKP states, scales exponentially with the number of modes since
\begin{align}
Z_\mathcal{W}(\projector{\phi})&=2^{m+1}\int d^{2m}\vec{\alpha}\abs{W_{\phi}(\vec{\alpha})}\\
&\leq2^{m+1}\left\lvert\int d^{2m}\vec{\alpha}W_{\phi}(\vec{\alpha})\right\rvert=2^{m+1}.
\end{align}

Therefore, a spanning set of product measurements is not always efficient, although the set of arbitrary product measurements is sufficient for performing efficient DFE of arbitrary product states:
\begin{corollary}
\label{cr:4}
If $\ket{\psi}$ is a product state, tomography of the neighborhood space generated with local operators $\mathcal{K}_{local}$, i.e. $N_k(\ket{\psi},\mathcal{K}_{local})$, where $\abs{\mathcal{K}_{local}}=O(\poly(m))$ and $k=O(1)$, can be done in a polynomial number of product measurements.
\end{corollary}
There exists an orthonormal basis of the neighborhood subspace where every basis state is a product state, and any $\ketbra{\psi_a}{\psi_b}$ is a projection to a single party pure state in at least $m-2k$ parties. Since only maximally $2k$ modes involve nontrivial measurements, $Z=2^{O(k)}=O(1)$. For example, our method is demonstrated experimentally to be effective in a multimode bosonic system on the subspace of maximally one photon among up to five modes~\cite{He2024}. This corresponds to $N_1(\ket{0}^{\otimes m},\{\hat{a}^\dagger_i\})$, which is $m+1$ dimensional. Off-diagonal terms are of the form of $\projector{0}^{\otimes m-1}\otimes\ketbra{1}{0}$ or $\projector{0}^{\otimes m-2}\otimes\ketbra{1}{0}\otimes\ketbra{0}{1}$, which were estimated by projecting $\projector{0}$ modes to vacuum and performing Wigner tomography in the remaining modes. This is listed as the second column in TABLE \ref{table:1} \& the third column in TABLE \ref{table:2}.

With some multimode controls, one can extend the set of DFE-efficient states further to include MPS:
\begin{corollary}
\label{cr:5}
If arbitrary quasilocal unitaries (i.e. arbitrary gates between any two subsystems) and single qubit measurements are allowed, tomography of any polynomial subspace spanned by MPSs with a fixed maximal bond dimension $k=O(1)$ can be performed efficiently with $O(1)$ ancilla qubits and $O(m)$ arbitrary two-qubit gates per measurement.
\end{corollary}
This follows from~\cite{PhysRevA.75.032311}, which showed that any MPS state with bond dimension $k$ can be efficiently converted to a product state with quasilocal unitaries, namely arbitrary unitaries between any specified qubit and a $k$-dimensional ancilla system. Hence, the quasilocal-DFE scale factor of MPS with a small bond dimension is $1$. Since the superposition of two MPS states of bond dimension $k$ can be written as a MPS state of bond dimension $2k$, we can perform DFE of it efficiently with $1+\lceil\log_2{k}\rceil$ ancilla qubits and arbitrary $(2+\lceil\log_2{k}\rceil)$-qubit unitaries. In particular,
\begin{equation}
\begin{split}
\mathcal{M}=&\left\{\left[\left(\prod_{i=1}^mU_i\right)^\dagger M\left(\prod_{i=1}^mU_i\right)\right]\right|M\in \mathcal{M}_{prod},\\
&\left.U_i\text{ is a unitary between }i^{th}\text{ qubit and the ancillas}\right\}.
\end{split}
\end{equation}
Here $\mathcal{M}_{prod}$ indicates the set of product measurements. From that, we can extract the off-diagonal components and construct the projected density operator efficiently. This is listed as the third column in TABLE \ref{table:1}. Compared to MPS tomography~\cite{cramer_efficient_2010,lanyon_efficient_2017}, our scheme is limited to predetermined reference states, but allows fidelity estimation even when the infidelity is large. In particular, even when the physical state requires intractably large bond dimensional MPO to describe, our scheme can still be used to determine the fidelity with the reference states.

Similarly, with arbitrary Gaussian unitaries, any pure Gaussian state in a multimode bosonic system can be efficiently converted to a vacuum state. Hence, we have the following, which is listed as the fourth column in TABLE \ref{table:2}:
\begin{corollary}
\label{cr:6}
If arbitrary Gaussian unitaries and single-mode measurements are allowed, DFE of arbitrary pure Gaussian states is efficient.
\end{corollary}

DEMESST can be extended to multiple base states, as long as DFE of any superposition of the base states can be performed efficiently. The base state condition is not bounded by limited entanglement: it only requires DFE of any superposition of the base states, which can be done efficiently if the base states are well-structured, as demonstrated by the example of stabilizer states. The relevant space would then be the combination of the individual neighborhoods of the base states.
\section{Discussion and Conclusion}
Our DEMESST framework has additional important features. 
First, DEMESST can \emph{self-verify} the polynomial neighborhood subspace assumption and, similarly, can determine the projected state even if a major part of the physical state is outside of the subspace. Second, DEMESST can work with a large class of states that can be highly entangled, as long as the DFE cost factor is tractable, without requiring any structure of the states. Furthermore, Generalized DFE can also be used to accommodate methods to determine some nonlinear properties of a density operator with joint measurement on multiple copies of the state. For example, determining purity $\tr(\rho^2)$ with joint Bell measurements of two copies of multiqubit states can be interpreted as performing the DFE of the SWAP operator with two copies of $\rho$ as the state to be sampled, since $\tr(\rho^2)=\tr(\mathrm{SWAP} \rho\otimes\rho)$.

In conclusion, we have presented a generalized version of DFE to minimize variance and shown an error bound for it. We have described DEMESST, a partial tomography of the neighborhood of a multipartite pure state with polynomial sampling complexity.
Future directions include proving or disproving efficient DFE of an arbitrary pure state with LOCC measurements, which would allow DEMESST of any polynomial subspace with a polynomial number of measurements (refer to Appendix \ref{sec:LOCC}), demonstrating or disproving a superpolynomial lower bound of cost for MPS, Gaussian states, and arbitrary pure states with more complicated measurements (refer to TABLES \ref{table:1} \& \ref{table:2}), which, if disproven, may expand the set of possible DFE-efficient states as base states we can use in practice.

\section{Acknowledgements}
We acknowledge support from the ARO(W911NF-23-1-0077), ARO MURI (W911NF-21-1-0325), AFOSR MURI (FA9550-19-1-0399, FA9550-21-1-0209, FA9550-23-1-0338), DARPA (HR0011-24-9-0359, HR0011-24-9-0361), NSF (OMA-1936118, ERC-1941583, OMA-2137642, OSI-2326767, CCF-2312755), NTT Research, Samsung GRO,  Packard Foundation (2020-71479), and the Marshall and Arlene Bennett Family Research Program. This material is based upon work supported by the U.S. Department of Energy, Office of Science, National Quantum Information Science Research Centers, and Advanced Scientific Computing Research (ASCR) program under contract number DE-AC02-06CH11357 as part of the InterQnet quantum networking project.
\newpage

\bibliography{main}
\newpage
\appendix
\section{Supplementary Materials for Section \ref{sec:gdfe}}
\subsection{DFE cost with POVM / unbounded measurements}\label{sec:povm_unbounded}
In general, one can obtain a Hoeffding's bound with POVM and unbound assigned values as long as the distribution of the measurement result is subgaussian for any physical state. The condition for a distribution of measurements to be subgaussian is to have a finite subgaussian norm \cite{Wainwright_2019}, which is defined as follows for quantum measurements:
\begin{equation}
\norm{X}_{\psi_2}=\inf_{c>0}\max_{\substack{\rho\\s.t. \sum_{i,j}p_i\tr(\rho\Lambda_j^{(i)})e^{\left(\lambda_j^{(i)}\right)^2/c^2}\leq2}}c,
\end{equation}
where $X$ is the random variable representing the measurement result, $\{\Lambda_j^{(i)}\}$ is the $i^{th}$ POVM, $p_i$ is the probability of choosing $i^{th}$ POVM, and $\lambda^{(i)}_j$ is the measurement result assigned to $\Lambda_j^{(i)}$. We can minimize over all possible distributions that realize an operator $\mathcal{O}$:
\begin{equation}
\norm{\mathcal{O}}_{\psi_2,\mathcal{L}}=\min_{\substack{\{\Lambda_j^{(i)}\}\in\mathcal{L},\{\lambda_j^{(i)}\},\{p_i\},C\\s.t. \sum_{i,j}p_i\lambda_j^{(i)}\Lambda_j^{(i)}=\mathcal{O}-CI}}\norm{X}_{\psi_2},
\end{equation}
where $\mathcal{L}$ is the set of available POVMs. For any $\mathcal{O}$ with finite $Z_{\mathcal{M}}(\mathcal{O})$, we have $Z_{\mathcal{M}}(\mathcal{O})\geq\norm{\mathcal{O}}_{\psi_2,\mathcal{L}}/\sqrt{\ln2}$, since any bounded random variable is also subgaussian. Conversely, with any finite $\norm{\mathcal{O}}_{\psi_2,\mathcal{L}}$, we have a Hoeffding bound of
\begin{equation}
\delta_t=\mathbb{P}_t\left(\abs*{\frac{1}{t}\sum_{j=1}^t X_j+C-\expval{\mathcal{O}}_\rho}\geq\epsilon\right)\leq2e^{-\Omega(t\epsilon^2\norm{\mathcal{O}}_{\psi_2,\mathcal{L}}^{-2})}.
\end{equation}
Hence, $\norm{\mathcal{O}}_{\psi_2,\mathcal{L}}$ is a generalized version of the DFE cost.
\subsection{Proofs of lemmas and theorem 1}\label{proofs}
\subsubsection{Lemma 2}
\begin{proof}
Proof for the `only if' statement:
\begin{equation}
\begin{aligned}
&i^c\ketbra{\psi_a}{\psi_b}+h.c.\\
=&2\cdot\frac{1}{2}\left(\ket{\psi_a}+i^{-c}\ket{\psi_b}\right)\left(\bra{\psi_a}+i^c\bra{\psi_b}\right)\\
&-\projector{\psi_a}-\projector{\psi_b},
\end{aligned}
\end{equation}
\begin{equation}
\begin{aligned}
&Z_\mathcal{M}(i^c\ketbra{\psi_a}{\psi_b}+h.c.)\\
\leq&2Z_\mathcal{M}\left(\frac{1}{2}\left(\ket{\psi_a}+i^{-c}\ket{\psi_b}\right)\left(\bra{\psi_a}+i^c\bra{\psi_b}\right)\right)\\
&+Z_\mathcal{M}(\projector{\psi_a})+Z_\mathcal{M}(\projector{\psi_b})\\
=&O(\poly(m)).
\end{aligned}
\end{equation}
Proof for the `if' statement:
\begin{equation}
\begin{aligned}
&Z_\mathcal{M}\left(\sum_a \alpha_a\ket{\psi_a}\sum_b \alpha_b^*\bra{\psi_b}\right)\\
\leq&\sum_{a>b}\left(Z_\mathcal{M}\left(\ketbra{\psi_a}{\psi_b}+h.c.\right)\abs{\text{Re}(\alpha_a\alpha_b^*)}\right.\\
&\left.+Z_\mathcal{M}\left(i\ketbra{\psi_a}{\psi_b}+h.c.\right)\abs{\text{Im}(\alpha_a\alpha_b^*)}\right)\\
&+\sum_aZ_\mathcal{M}(\projector{\psi_a})\abs{\alpha_a}^2\\
=&O(\poly(m)).
\end{aligned}
\end{equation}
\end{proof}
\subsubsection{Lemma 3}
\begin{proof}
Each $i^c\ketbra{\psi_a}{\psi_b}+i^{-c}\ketbra{\psi_b}{\psi_a}$ can be determined to $\epsilon'$ precision with $1-\delta'$ probability with $t=\frac{Z^2}{2\epsilon'^2}\ln\frac{2}{\delta'}=O\left(\poly(m,\ln\delta',\epsilon'^{-1})\right)$ measurements. Hence, setting $\delta'=\delta/(\dim{\mathcal{S}})^2$ and $\epsilon'=\epsilon/\dim{\mathcal{S}}$, we can determine $\rho_\mathcal{S}$ with $(\dim{\mathcal{S}})^2t=O\left(\poly(m,\ln\delta,\epsilon^{-1})\right)$ measurements.
\end{proof}
\subsubsection{Theorem 1}
\begin{proof}
The diagonal elements of the density matrix on the neighborhood basis can be estimated by rotating the basis state to the base state followed by DFE of the base state. To estimate the off-diagonal elements, we apply controlled unitaries on the subsystems with differing unitaries on both sides, and the required elements can be calculated through the difference in the results. For example, the real part of $\expval{U_a^\dagger\mathcal{O}U_b}$ can be calculated by performing a Hadamard test on controlled unitaries $U_a$ and $U_b$:
\begin{equation}
\begin{split}
\text{Re}\left(\expval{U_a^\dagger\mathcal{O}U_b}\right) ={}& \frac{1}{2}\left(\expval*{\left(\frac{U_a+U_b}{\sqrt{2}}\right)^\dagger\mathcal{O}\left(\frac{U_a+U_b}{\sqrt{2}}\right)}\right.\\
&-\left.\expval*{\left(\frac{U_a-U_b}{\sqrt{2}}\right)^\dagger\mathcal{O}\left(\frac{U_a-U_b}{\sqrt{2}}\right)}\right).
\end{split}
\end{equation}
By performing up to $2k$ Hadamard tests, followed by the original measurement, we can perform DFE for operators within the neighborhood of any DFE-efficient state.

Any pure state within the neighborhood can be written as a superposition of $\prod_{i=1}^k U_{a_i}\ket{\psi}$ with coefficients $\leq \norm{G^+}_2$, hence $Z_\mathcal{M}=O(\poly(m,\norm{G^+}_2))$.
\end{proof}
\section{LOCC Conjecture}\label{sec:LOCC}
To systematically generate a DFE of any arbitrary pure state with only LOCC measurements, we introduce a mathematical construct, which we will call  ``measurement contrast", that characterizes the sampling complexity of DFE of a Hermitian operator $\mathcal{O}$, given a set of possible bounded measurement operators. Here, LOCC measurement refers to any measurement that involves only LOCC operations.
\begin{definition}
Define the traceless operator $\widetilde{\mathcal{O}}$
\begin{equation}
\widetilde{\mathcal{O}}=\mathcal{O}-\frac{\tr(\mathcal{O})}{D}I,
\end{equation}
and the ``effective contrast" $Y$
\begin{equation}Y(\mathcal{O},M,\varsigma)=
\begin{cases}
    0 &\textbf{ if }\tr(M\tilde{\varsigma})>0,\\
    \frac{\tr\left(M\widetilde{\mathcal{O}}\right)}{\lambda_{\max}(M)-\lambda_{\min}(M)} &\textbf{ if }\tr(M\tilde{\varsigma})\leq 0,
\end{cases} 
\end{equation}
where $D$ is the dimension of the overall Hilbert space, $\lambda_{\max}$ and $\lambda_{\min}$ are the maximum and minimum value assigned to any measurement result, and $\tilde{\varsigma}=\varsigma-\frac{I}{D}\tr(\varsigma)-\frac{\widetilde{\mathcal{O}}}{\norm{\widetilde{\mathcal{O}}}_F^2}\tr\left(\varsigma\widetilde{\mathcal{O}}\right)$. Here $M$ is a measurement operator, while $\varsigma$ is some Hermitian operator that will be determined later.
For any Hermitian operator $\mathcal{O}$ and the set of measureable operators $\mathcal{M}$, we define the measurement contrast to be
\begin{equation}\label{eq:1}
Y_{\mathcal{M}}(\mathcal{O})=\min_{\varsigma\in\mathcal{H}}\max_{M\in\mathcal{M}}Y(\mathcal{O},M,\varsigma),
\end{equation}
where $\mathcal{H}$ is the set of all Hermitian operators. If this value is positive, it implies that for any $\varsigma$, there exists some $M\in\mathcal{M}$ such that $Y(\mathcal{O},M,\varsigma)\geq Y_{\mathcal{M}}(\mathcal{O})$.
\end{definition}
This measurement contrast, which we will demonstrate in the following theorem, indicates that there exists a random ensemble of measurements from $\mathcal{M}$ such that the expected value is proportional to the expected value of $\mathcal{O}$ with a proportional constant of at least $Y_\mathcal{M} (\mathcal{O})$.
\begin{theorem}\label{th:2}
Given a set of possible measurable operators $\mathcal{M}$, if $Y_\mathcal{M}(\mathcal{O})>0$, then $Z_\mathcal{M}(\mathcal{O})\leq \norm{\widetilde{\mathcal{O}}}_F^2/Y_\mathcal{M}(\mathcal{O})$.
\end{theorem}
\begin{proof}
In this proof, we use an iterative approach to generate an ensemble of measurements to approximate $\mathcal{O}$ with a systematic error that limits to 0 as the iteration number increases. For an ensemble of $M_i$ and the corresponding probability distribution $p_i^{(n)}$ in the $n$-th iteration, we define
\begin{equation}
\left\{\begin{aligned}
&\widetilde{M}_i=\frac{M_i-\frac{\tr(M_i)}{D}I}{\lambda_{\max}(M)-\lambda_{\min}(M)},\\
&\Delta_i=\widetilde{M}_i-\frac{\tr\left(\widetilde{M}_i\widetilde{\mathcal{O}}\right)}{\norm{\widetilde{\mathcal{O}}}_F^2}\widetilde{\mathcal{O}}=\widetilde{M}_i-\frac{Y(\mathcal{O},M_i,\varsigma_{i-1})}{\norm{\widetilde{\mathcal{O}}}_F^2}\widetilde{\mathcal{O}},\\
&\varsigma_n=\sum_ip^{(n)}_i\Delta_i.\end{aligned}\right.
\end{equation}
Here $\Delta_i$ is the deviation for the measurement $M_i$ from approximating $\mathcal{O}$, while $\varsigma_n$ is the deviation of the ensemble at step $n$. We can see that $\Delta$ and $\varsigma$ have zero traces and zero overlaps with $\mathcal{O}$. Moreover,
\begin{equation}\norm{\Delta_i}_F\leq\norm{\widetilde{M}_i}_F<\sqrt{D}.\end{equation}
Now we discuss how we generate the ensemble of measurement $\{M_i\}$ and the corresponding probability $p^{(n)}_i$ for each operator $M_i$ iteratively. First, we set $\varsigma_0=0$. From Eq.~\eqref{eq:1}, there exists $M_1$ with $\tr\left(\widetilde{M}_1\widetilde{\mathcal{O}}\right)\geq Y_\mathcal{M}(\mathcal{O})$. We set $p_1^{(1)}=1$ accordingly, leading to the deviation $\varsigma_1=\Delta_1$. 

Then, after finishing the $n$-th iteration, we will update $\{M_i\}$ and $p^{(n+1)}_i$ in the following way. Given $\varsigma_n$, from Eq.~\eqref{eq:1}, there exists $M_{n+1}$ such that $0\geq\tr(\varsigma_nM_{n+1})=\tr(\varsigma_n\Delta_{n+1})$ with $\tr\left(\widetilde{M}_{n+1}\widetilde{\mathcal{O}}\right)\geq Y_\mathcal{M}(\mathcal{O})$. Set $p_{n+1}^{(n+1)}=\frac{\norm{\Delta_{n+1}}_F^2}{\norm{\Delta_{n+1}}_F^2+\norm{\varsigma_n}_F^2}$ and $p_i^{(n+1)}=\left(1-p_{n+1}^{(n+1)}\right)p_i^{(n)}$, then the deviation will be $\varsigma_{n+1} = p_{n+1}^{(n+1)}\Delta_{n+1}+\left(1-p_{n+1}^{(n+1)}\right)\varsigma_n$, which satisfies $\norm{\varsigma_{n+1}}_F^{-2}\geq\norm{\Delta_{n+1}}_F^{-2}+\norm{\varsigma_n}_F^{-2}$. From this, we have a bound of the Frobenius norm of this new $\varsigma$:
\begin{equation}
\norm{\varsigma_{n+1}}_F\leq\left(\sum_i\norm{\Delta_i}^{-2}\right)^{-1/2}<\sqrt{D/(n+1)}.
\end{equation}
In the limit of $n$ going to $\infty$, $\lim_{n\rightarrow\infty}\norm{\varsigma_n}_F\leq\lim_{n\rightarrow\infty}\sqrt{D/n}=0$. Therefore,
\begin{equation}
\lim_{n\rightarrow\infty}\left(\sum_ip_i^{(n)}\widetilde{M}_i-\sum_ip_i^{(n)}\frac{\tr\left(\widetilde{M}_i\widetilde{\mathcal{O}}\right)}{\norm{\widetilde{\mathcal{O}}}_F^2}\widetilde{\mathcal{O}}\right)=0.
\end{equation}
Hence, we demonstrated there exists an ensemble to approximate $\mathcal{O}$:
\begin{equation}Z\sum_ip^{(n)}_i\widetilde{M}_i+CI=\mathcal{O},
\end{equation}
with
\begin{equation}Z=\left(\sum_ip_i^{(n)}\frac{\tr\left(\widetilde{M}_i\widetilde{\mathcal{O}}\right)}{\norm{\widetilde{\mathcal{O}}}_F^2}\right)^{-1}\leq\frac{\norm{\widetilde{\mathcal{O}}}_F^2}{Y_\mathcal{M}(\widetilde{\mathcal{O}})},
\end{equation}
and
\begin{equation}C=\frac{\tr\left(\mathcal{O}\right)}{D}.
\end{equation}
\end{proof}
With this framework, we conjecture the following, which, when combined with Theorem \ref{th:2}, is a sufficient condition for DFE with only LOCC measurements to be efficient for any pure state.
\begin{conjecture}\label{cj:1}
For any Hermitian operator $\mathcal{O}$ in an $m$-qubit system, there exists an LOCC measurement $M$ with $\lambda_{\max}=1$ and $\lambda_{\min}=-1$ such that
\begin{equation}
\tr{(M)}=0,
\end{equation}
\begin{equation}
\tr{\left(M\widetilde{\mathcal{O}}\right)}\geq\left(1-2^{-m}\right)^{-1}\norm{\widetilde{\mathcal{O}}}_2.
\end{equation}
\end{conjecture}
\begin{corollary}
If Conjecture \ref{cj:1} holds, then for any pure state $\ket{\psi}$, $Z_{LOCC}(\projector{\psi})\leq2-2^{1-m}$.
\end{corollary}
\begin{proof}
Consider $\mathcal{O}=\widetilde{\mathcal{O}}=\left(1-2^{-m}\right)^{-1}\projector{\psi}-\left(2^m-1\right)^{-1}I-\epsilon\varsigma$, where $\varsigma$ is a traceless Hermitian operator satisfying $\varsigma\ket{\psi}=0$ and $\epsilon$ is a positive number satisfying $\epsilon\norm{\varsigma}_2\leq1-\left(2^m-1\right)^{-1}$, we have $\norm{\widetilde{\mathcal{O}}}_2=1$. From conjecture \ref{cj:1}, there exists $M$ such that $\tr(M)=0$ and
\begin{equation}
\begin{split}
\tr{\left(M\left(\widetilde{\mathcal{O}}+\epsilon\varsigma\right)\right)}&=\frac{\braketOP{\psi}{M}{\psi}}{1-2^{-m}}-\frac{\tr{(M)}}{2^m-1}\\
&\leq\frac{\norm{\widetilde{\mathcal{O}}}_2}{1-2^{-m}}\leq\tr\left(M\widetilde{\mathcal{O}}\right).
\end{split}
\end{equation}
Hence, we have
\begin{equation}
\tr{(M\varsigma)}\leq 0.
\end{equation}
\begin{equation}
\begin{split}
\frac{\tr\left(M\left(\projector{\psi}-2^{-m}I\right)\right)}{\lambda_{\max}(M)-\lambda_{\min}(M)}&=\frac{1-2^{-m}}{2}\tr{\left(M\left(\widetilde{\mathcal{O}}+\epsilon\varsigma\right)\right)}\\
&\geq\frac{1}{2}+\frac{1-2^{-m}}{2}\epsilon\tr\left(M\varsigma\right)\\
&>\frac{1}{2}\left(1-\epsilon\norm{\varsigma}_1\right).
\end{split}
\end{equation}
Therefore, for any positive $\epsilon$ we have $Y(\projector{\psi},M,\varsigma)>\frac{1}{2}\left(1-\epsilon\norm{\varsigma}_1\right)$. Since this holds for any sufficiently small $\epsilon$, we must have $Y_{LOCC}(\projector{\psi})\geq 1/2$. From Theorem \ref{th:2}, we thus have
\begin{equation}
Z_{LOCC}(\projector{\psi})\leq 2\norm{\projector{\psi}-2^{-m}I}_F^2=2-2^{1-m}.
\end{equation}
\end{proof}
For example, the DFE cost of any two-qubit entangled pure state with LOCC measurements is at most the same as that of a Bell pair: Any two-qubit pure state is equivalent to the Schmidt form up to single-qubit unitaries, hence without loss of generality let $\ket{\psi}=\sqrt{\lambda}\ket{00}+\sqrt{1-\lambda}\ket{11}$, with $\lambda\geq1/2$. To perform DFE of such state with LOCC measurements, we can measure a random Pauli on a random qubit, followed by measuring the other qubit along the projected state, i.e.
\begin{equation}
\frac{2\hat{M}_{z}+\hat{M}_{x1}+\hat{M}_{x2}+\hat{M}_{y1}+\hat{M}_{y2}}{6}=\frac{2\hat{M}_{\psi}+I}{3},
\end{equation}
where $\hat{M}_{z}=\projector{00}+\projector{11}$, $\hat{M}_{x1}=\projector{+}\otimes\left(\sqrt{\lambda}\ket{0}+\sqrt{1-\lambda}\ket{1}\right)\left(\sqrt{\lambda}\bra{0}+\sqrt{1-\lambda}\bra{1}\right)+\projector{-}\otimes\left(\sqrt{\lambda}\ket{0}-\sqrt{1-\lambda}\ket{1}\right)\left(\sqrt{\lambda}\bra{0}-\sqrt{1-\lambda}\bra{1}\right)$, and similar for $\hat{M}_{y1},\hat{M}_{x2},\hat{M}_{y2}$.
\end{document}